\documentclass[]{article}
\usepackage[affil-it]{authblk}

\usepackage{fullpage}
\usepackage{graphicx}
\usepackage{pdflscape}
\usepackage{subcaption}
\usepackage{amsthm}
\usepackage{amsmath, amssymb}
\usepackage{amsfonts}
\usepackage{booktabs}
\usepackage[]{natbib} %
\usepackage[usenames,dvipsnames]{color}
\usepackage[linesnumbered, noend]{algorithm2e}
\usepackage{tikz}
\usepackage{tikzsymbols}
\usetikzlibrary{arrows}
\usetikzlibrary{calc}

\newtheorem{theorem}{Theorem}

\newtheorem{observation}{Observation}
\DeclareGraphicsExtensions{.png,.pdf,.jpg}

\definecolor{purple}{RGB}{153,50,204}

\newcommand{\ABP}{ABP}

\usepackage{hyperref}
\begin{document}

\title{A note on computational approaches for the antibandwidth problem}

\author[1]{Markus Sinnl\thanks{markus.sinnl@jku.at}}

\affil[1]{Department of Statistics and Operations Research, Faculty of Business, Economics and Statistics, University of Vienna, Vienna, Austria \newline
	Institute of Production and Logistics Management, Johannes Kepler University Linz, Linz, Austria}

\title{A note on computational approaches for the antibandwidth problem}

\date{}
\maketitle

\begin{abstract}
In this note, we consider the antibandwidth problem, also known as dual bandwidth problem, separation problem and maximum differential coloring problem. Given a labeled graph (i.e., a numbering of the vertices of a graph), the antibandwidth of a node is defined as the minimum absolute difference of its labeling to the labeling of all its adjacent vertices. The goal in the antibandwidth problem is to find a labeling maximizing the antibandwidth. The problem is NP-hard in general graphs and has applications in diverse areas like scheduling, radio frequency assignment, obnoxious facility location and map-coloring.

There has been much work on deriving theoretical bounds for the problem and also in the design of metaheuristics in recent years. However, the optimality gaps between the best known solution values and reported upper bounds for the HarwellBoeing Matrix-instances, which are the commonly used benchmark instances for this problem, are often very large (e.g., up to 577\%). Moreover, only for three of these 24 instances, the optimal solution is known, leading the authors of a state-of-the-art heuristic to conclude "HarwellBoeing instances are actually a challenge for modern heuristic methods". The upper bounds reported in literature are based on the theoretical bounds involving simple graph characteristics, i.e., size, order and degree, and a mixed-integer programming (MIP) model.

We present new MIP models for the problem, together with valid inequalities, and design a branch-and-cut algorithm and an iterative solution algorithm based on them. These algorithms also include two starting heuristics and a primal heuristic. We also present a constraint programming approach, and calculate upper bounds based on the stability number and chromatic number. Our computational study shows that the developed approaches allow to find the proven optimal solution for eight instances from literature, where the optimal solution was unknown and also provide reduced gaps for eleven additional instances, including improved solution values for seven instances, the largest optimality gap is now 46\%.
\end{abstract}


\section{Introduction and motivation \label{sec:intro}}

Graph labeling problems are an important class of problems, which have been studied since the 1960s. In such problems, we are given a graph and we want to find a labeling (i.e., a numbering of its vertices), such that a given objective function is optimized. Problems in this class include the \emph{bandwidth problem} \citep{cuthill1969reducing,caprara2005laying} and variants of it like \emph{cyclic bandwidth} \citep{rodriguez2015tabu}, the \emph{linear arrangement problem} \citep{caprara2011decorous,rodriguez2008effective} and the \emph{cutwidth problem} \citep{marti2013branch}, see also the surveys \citep{diaz2002survey,gallian2009dynamic}. In this work, we consider the \emph{antibandwidth problem} (\ABP), also known as \emph{dual bandwidth problem} \citep{yixun2003dual}, \emph{separation problem} \citep{miller1989separation} and \emph{maximum differential coloring problem} \citep{bekos2014note}. The \ABP\ is NP-hard in general graphs and has applications in scheduling \citep{leung1984some}, radio frequency assignment \citep{hale1980frequency}, obnoxious facility location \citep{cappanera1999survey} and map-coloring \citep{gansner2010gmap}.

\paragraph{Problem definition}
Let $G=(V,E)$ be a graph, where $V$ is the set of vertices and $E$ is the set of edges, and let $n=|V|$ and $m=|E|$
A labeling $f$ of the vertices is a bijection $V\rightarrow \{1,\ldots,n\}$, i.e., each vertex $i \in V$ gets a unique label $f(i) \in \{1,\ldots,n\}$. For a graph $G$ and a labeling $f$, the antibandwidth $AB_f(G)$ is
\begin{align*}
AB_f(G)=\min\{AB_f(i): i \in V \} 
\end{align*}
where 
\begin{align*}
AB_f(i)=\min\{|f(i)-f(i')|: \{i,i'\} \in E \} 
\end{align*}

is the minimum bandwidth of a vertex $i \in V$ (we will also call this antibandwidth of $i$). Let $\mathcal F(G)$ denote all labelings of $G$. The \ABP\ consists of finding a labeling $f^*$ that maximizes $AB_f(G)$  and the corresponding value $AB_{f^*}(G)$ is called antibandwidth $AB(G)$ of the graph, i.e.,
\begin{align*}
AB(G)=\max_{f \in \mathcal F} AB_f(G)
\end{align*}

For ease of readability, we write $AB_f$ instead of $AB_f(G)$ in the following.
For later use, for two numbers (labels) $a, a'$, let $d(a,a')=|a-a'|$ and for a set of numbers $A$, let $d(a,A)=\min_{a' \in A}|a-a'|$; for a vertex $i \in V$, let $\delta(i)$ denote its degree and $\Delta^+=\max_{i \in V} \delta(i)$, $\Delta^-=\min_{i \in V} \delta(i)$, denote the maximum, resp., minimum degree of a vertex in the considered graph.
Figure \ref{fig:example} shows an exemplary instance of the \ABP\ together with an optimal labeling.

\tikzstyle{vertex}=[circle,fill=black!15,minimum size=20pt,inner sep=0pt]
\tikzstyle{edge} = [draw,thick,-]
\begin{figure}[h!tb]
\begin{subfigure}[b]{.5\linewidth}
\centering
\begin{tikzpicture}[scale=1.5]
\foreach \pos/\name/\type in {{(0,2)/A/vertex}, {(1,2)/B/vertex}, {(2,2)/C/vertex},{(0,1)/D/vertex},{(1,1)/E/vertex},{(2,1)/F/vertex},{(0,0)/G/vertex},{(1,0)/H/vertex},{(2,0)/I/vertex}}
        \node[\type] (\name) at \pos {$\name$};
\foreach \source/ \dest in  
    {A/B,B/C,A/D,D/E,B/E,C/F,E/F,D/G,G/H,E/H,H/I,F/I}
    \path[edge] (\source) --  (\dest);
\end{tikzpicture}
\caption{Instance}\label{fig:instance}
\end{subfigure}%
\begin{subfigure}[b]{.5\linewidth}
\centering
\begin{tikzpicture}[scale=1.5]
\foreach \pos/\name/\type in {{(0,2)/5/vertex}, {(1,2)/1/vertex}, {(2,2)/6/vertex},{(0,1)/2/vertex},{(1,1)/7/vertex},{(2,1)/3/vertex},{(0,0)/8/vertex},{(1,0)/4/vertex},{(2,0)/9/vertex}}
        \node[\type] (\name) at \pos {$\name$};
\foreach \source/ \dest in  
    {A/B,B/C,A/D,D/E,B/E,C/F,E/F,D/G,G/H,E/H,H/I,F/I}
    \path[edge] (\source) --  (\dest);
\end{tikzpicture}

\caption{Optimal labeling}\label{fig:solution}
\end{subfigure}
\caption{Instance $G$ and optimal solution, $AB(G)=3$, as for edge $\{C,F\}$, we have $|f(C)-f(F)|=|6-3|=3$ (edge $\{E,H\}$ also gives value three). \label{fig:example}}
\end{figure}
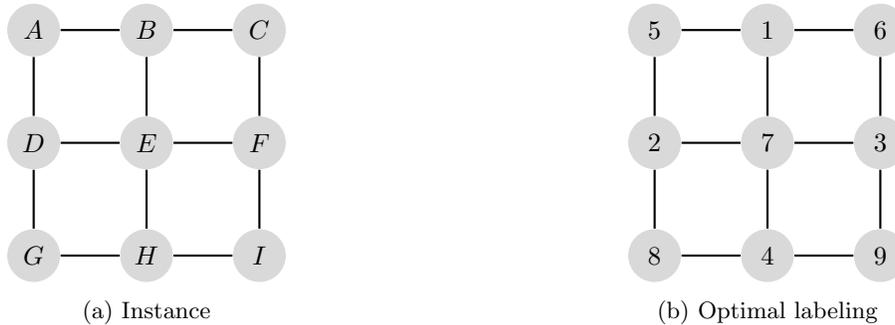

%

\paragraph{Previous work}

In \cite{miller1989separation,yixun2003dual} various theoretical bounds for general graphs based on graph parameters like size, order, degree, stability number and chromatic number are presented. For certain classes of graphs like Hamming graphs \citep{dobrev2013antibandwidth}, hypercubes \citep{raspaud2009antibandwidth,wang2009explicit}, complete $k$-ary trees \citep{calamoneri2009antibandwidth}, caterpillars and spiders \citep{bekos2013maximum,bekos2014note} there exist tighter bounds and/or exact algorithms.

For general graphs, a variety of (meta-)heuristic approaches exist: \cite{bansal2011memetic} proposed a memetic algorithm, \cite{duarte2011grasp} develops a generalized randomized adaptive search procedure with path relinking, \cite{lozano2012variable} presented a variable neighborhood search and \cite{scott2014level} designed a hill-climbing algorithm. \cite{duarte2011grasp} also introduced a mixed-integer programming (MIP) model for the exact solution of the \ABP, see Section \ref{sec:previous} for the model.

\paragraph{Contribution and outline}

While there has been much work on deriving theoretical bounds for the problem and also in the design of metaheuristics, the optimality gaps between the best known solution values and reported upper bounds for the HarwellBoeing Matrix-instances, which are the commonly used benchmark instances for this problem, are often very large (e.g., up to 577\%, see Table \ref{ta:main}). Only for three of the 24 instances, the optimal solution is known. Aside from the upper bounds provided by the MIP of \cite{duarte2011grasp}, the upper bounds reported in literature are based on the theoretical bounds involving simple graph characteristics, i.e., size, order and degree, leading to the conclusion "On the contrary, the CBT, Hamming
and HarwellBoeing instances are actually a challenge for modern heuristic methods" in \cite{lozano2012variable}, which presents a the state-of-the-art heuristic for the problem\footnote{for the instance sets CBT, which are complete binary trees and Hamming, which are Hamming graphs, graph-specific algorithms producing the optimal solution are known, see the \emph{Previous Work} paragraph above}.

In this note, we present two new MIP formulations for the problem and design a branch-and-cut algorithm and an iterative solution algorithm based on them. The branch-and-cut algorithms include valid inequalities, two starting heuristics and a primal heuristic. We also calculate bounds by using the stability number and chromatic numbers (these calculations are also done using MIPs to solve the associated NP-hard problems).  
The developed approaches and calculations allow to find the proven optimal solution for eight instances, where the optimal solution was not known, and reduced gaps for eleven additional instances, including seven improved solution values. The results reveal that the heuristics from literature presented for this problem actually work quite well, and the large optimality gaps reported so far are mainly caused by weak upper bounds.

In Section \ref{sec:previous} we recall the theoretical bounds known for the problem and also the MIP approach of \cite{duarte2011grasp}, and also discuss calculation of the stability number and chromatic number. In Section \ref{sec:mip} we present our new MIP models and also describe further details of our branch-and-cut algorithm and the iterative solution algorithm, including valid inequalities and heuristics. Section \ref{sec:cp} contains our constraint programming formulation.
Section \ref{sec:comp} details the obtained computational results, and Section \ref{sec:concl} concludes the paper.

\section{Upper bounds for the ABP \label{sec:previous}}

The following graph theoretic bounds are known.

\begin{theorem}[\cite{miller1989separation,yixun2003dual}]\label{thm:bounds}
Let $G$ be a connected graph, $\alpha(G)$ be the stability number of $G$ and $\chi(G)$ be the chromatic number of the graph. Then the following holds
\begin{enumerate}
\item $AB(G)\leq \min \Big\{\lfloor \frac{n-\Delta^-+1}{2}  \rfloor, n-\Delta^+  \Big\}$
\item $AB(G)\leq \lfloor n-\frac{\sqrt{8m+1}-1}{2} \rfloor$
\item $AB(G)\leq \alpha(G)$
\item $AB(G)\leq \lfloor \frac{n-1}{\chi(G)-1} \rfloor$
\end{enumerate}
\end{theorem}

Note that in previous work presenting heuristic approaches for the problem, aside from using the bound provided by the MIP in \cite{duarte2011grasp} (see below), only the first two bounds stated in Theorem \ref{thm:bounds} have been used to assess the quality of the generated heuristic solutions.

To calculate the bounds 3. and 4. in Theorem \ref{thm:bounds}, one needs to calculate the stability number $\alpha(G)$, resp., the chromatic number $\chi(G)$, i.e., one needs to solve the NP-hard \emph{(maximum) stable set problem (SSP)} (also known as \emph{independent set problem}), resp., \emph{(minimum) graph-coloring problem (GCP)}. Both problems are well-studied in literature and there are many different (exact and heuristic) solution approaches for it, see, e.g., the tutorial \citep{rebennack2012tutorial} and the surveys \citep{galinier2006survey,malaguti2010survey} for more details on these problems. For our purposes to calculate valid bounds for the ABP, we need the exact solution value (or the value of a relaxation). To calculate these values, we used standard MIP-models for both problems: For the SSP (see, e.g., \cite{rebennack2012tutorial}), let binary variable $x_i=1$, iff vertex $i \in V$ is in the stable set. The problem can be formulated as follows.

\begin{align}
\alpha(G)=\max_{x \in \{0,1\}^{|V|}} \Big\{\sum_{i  \in V} x_i : x_i +x_{i'} \leq 1, \forall \{i,i'\} \in E \Big\}. \label{eq:SSP} \tag{SSP}
\end{align}

For the GCP (see, e.g., \cite{mendez2006branch,mendez2008cutting}), let binary variable $x^c_i=1$, iff vertex $i \in V$ gets color $c \in \{1,\ldots,|V|\}$ in the solution, and let binary variable $w_c=1$, iff color $c \in \{1,\ldots,|V|\}$ is used in the solution. The problem can be formulated as follows. 

\begin{align*}
\chi(G)=\min_{x \in \{0,1\}^{|V|}} \Big\{\sum_{c \in \{1,\ldots,|V|\}} w_c : x^c_i +x^c_{i'} \leq w_c, \forall \{i,i'\} \in E , \forall c \in \{1,\ldots,|V|\}, 
\sum_{c \in \{1,\ldots,|V|\}} x_i^c=1, \forall i \in V \Big\}.
\end{align*}

To speed-up computation, instead of $c \in \{1,\ldots,|V|\}$, we use $c \in \{1,\ldots,|UB(\chi(G))|\}$, where $UB(\chi(G))$ is the value of an upper-bound for $\chi(G)$ obtained by a simple greedy heuristic \citep{leighton1979graph} for the GCP, the heuristic solution inducing this upper bound value is also given as starting solution to the MIP-solver.


\subsection{Mixed-Integer Programming approach of \cite{duarte2011grasp}}

In \cite{duarte2011grasp}, the following MIP-model based on a big-M formulation is presented. Let binary variables $x_i^\ell$ take the value one if and only if vertex $i$ gets label $\ell$ (i.e., $f_i=\ell$). The following set of assignment constraints \eqref{eq:a-sum} and \eqref{eq:b-sum} make sure that every vertex gets an unique labeling.

\begin{align} 
\sum_{i \in V } x^\ell_i&=1 &\quad \forall \ell \in\{1,\ldots,|V|\} \label{eq:a-sum} \tag{VERTICES}\\
\sum_{\ell \in\{1,\ldots,|V|\} } x^\ell_i&=1 &\quad \forall i \in V. \label{eq:b-sum} \tag{LABELS} 
\end{align} 

Let integer variables $l_i\in \{1,\ldots,|V|\}$ indicate the labeling of vertex $i \in V$. The $x$-variables and $l$-variables can be linked with the following set of constraints

\begin{align} 
\sum_{1\leq \ell \leq |V|} \ell x^\ell_i&=l_i &\quad \forall i \in\{1,\ldots,|V|\}. \label{eq:link} \tag{LINK}  
\end{align} 

Finally, let binary variables $y_{ii'}$ and $z_{ii'}$ for $\{i,i'\} \in E$ indicate whether $i$ has a smaller label than $i'$; if it has a smaller label, then $y_{ii'}=1$, otherwise $z_{ii'}=1$ (one could get rid of one set of these variables, but we want to follow \citep{duarte2011grasp} exactly), and let variable $b$ measure the value of the antibandwidth of the solution. The \ABP\ can than be formulated as follows (denoted as $(F_{lit})$).

\begin{align}
\max \quad & b & \notag \\
\eqref{eq:a-sum} & , \eqref{eq:b-sum}, \eqref{eq:link} & \notag \\
b - (l_i-l_{i'}) -2(|V|-1) y_{ii'}&\leq 0 & \forall \{i,i'\} \in E \label{eq:obj1} \tag{OBJ-1} \\
b - (l_{i'}-l_{i}) -2(|V|-1) z_{ii'} &\leq 0 & \forall \{i,i'\} \in E \label{eq:obj2} \tag{OBJ-2}\\
y_{ii'}+z_{ii'} & =1  & \forall \{i,i'\} \in E\label{eq:obj3} \tag{OBJ-3} \\
x^\ell_i &\in\{0,1\} & \forall i \in V, \forall \ell \in \{1,\ldots, |V| \} \notag \\
l_i & \in \{1,\ldots,|V|\} & \forall i \in V \notag \\
y_{ii'},z_{ii'} & \in\{0,1\}   & \forall \{i,i'\}\in E \notag
\end{align}

Constraints \eqref{eq:obj1}, \eqref{eq:obj2}, \eqref{eq:obj3}
model $b \leq |l_i-l_{i'}|$, $\forall \{i,i'\} \in E$ in a big-M-constraint style and ensure that $b$ correctly measures the antibandwidth of the solution indicated by the selected $l$ (resp., $x$)-variables: For $\{i,i'\} \in E$, suppose vertex $i$ has a smaller label than $i'$. Hence $b$ can be at most $d_{i'}-d_i$. Since the objective functions maximizes, $y_{ii'}$ will take the value one and $z_{ii'}$ will take the value zero, resulting in $b - (l_i-l_{i'}) \leq 2(|V|-1)$ for \eqref{eq:obj1} and $b - (l_{i'}-l_{i}) \leq 0$ for \eqref{eq:obj2}, which ensures that $b \leq l_{i'}-l_{i}$. The case for $i$ having a larger label than $i'$ works analogously. 
The resulting model has $O(|V|^2)$ variables and $O(|E|)$ constraints.

\section{New Mixed-Integer Programming approaches \label{sec:mip}}

\subsection{New formulation $(F)$}

As a first way to improve formulation $(F_{lit})$, one can downlift the coefficients $2(|V|-1)$ in \eqref{eq:obj1}, \eqref{eq:obj2} to $(|V|-1)+UB$ where $UB$ in any valid upper bound to \ABP\ for the considered instance. This follows from the fact, that $b\leq UB$ and $|l_i-l_{i'}| \leq |V|-1$ for any valid labeling. However, the problem can also be formulated without such big-M-constraints and variables $l$, $y$ and $z$, as shown next (denoted as formulation $F$).

\begin{align}
\max \quad & b & \notag \\
\eqref{eq:a-sum} & , \eqref{eq:b-sum} & \notag \\
b - \sum_{1 \leq \ell' \leq |V|} d(\ell,\ell')(x^{\ell'}_i+x^{\ell'}_{i'}) &\leq 0 &\quad \forall \ell \in \{1,\ldots, |V|\}, \forall \{i,i'\} \in E \label{eq:objn} \tag{OBJ-N}\\ 
x^\ell_i &\in\{0,1\} & \forall i \in V, \forall \ell \in \{1,\ldots, |V| \} \notag
\end{align}

Constraints \eqref{eq:objn} model $b \leq |l_i-l_{i'}|$, $\forall \{i,i'\} \in E$ and ensure that $b$ correctly measures the antibandwidth of the solution indicated by the selected  $x$-variables: For $\{i,i'\} \in E$, let $\ell(i)$ and $\ell(i')$ be the labels indicated by the values of $x^\ell_i$ and $x^\ell_{i'}$. For $\ell=\ell(i)$ the constraint \eqref{eq:objn} reads $b \leq d(\ell(i),\ell(i'))=|\ell(i)-\ell(i')|$, which is exactly as desired, the case for $\ell=\ell(i')$ works analogously. For $\ell\neq\ell(i),\ell(i')$, the constraint \eqref{eq:objn} reads $b \leq d(\ell,\ell(i))+d(\ell,\ell(i'))=|\ell-\ell(i)|+|\ell-\ell(i')|$, and due to the triangle inequality $|\ell-\ell(i)|+|\ell-\ell(i')|\geq |\ell(i)-\ell(i')|$, thus the constraint remains valid also in these cases. Formulation $(F)$ has $O(|V|^2)$-variables and $O(|V||E|)$-constraints.
Given a valid upper bound $UB$, each coefficient $d(\ell,\ell')>UB$ in constraints \eqref{eq:objn} can be downlifted to $UB$, clearly the constraints remain valid.
Moreover, constraints \eqref{eq:objn} are actually a special case of the following set of clique-based constraints.

\begin{observation}
Let $C\subseteq V$ be a set of vertices forming a clique in $G$ and let $L$ be a set of labels with $|L|=|C|-1$. Then inequalities
\begin{equation}
b - \sum_{c \in C} \sum_{1 \leq \ell \leq |V|} d(\ell,L) x^{\ell}_c \leq 0 \label{eq:cliquen} \tag{CLIQUE-N}
\end{equation}
are valid for $(F)$.
\end{observation}

\begin{proof}
For any labeling, at least one of the $x^\ell_i$ variables (with $x^\ell_i=1$ in this labeling) in \eqref{eq:cliquen} has a positive coefficient, as there are $|C|$ vertices in the clique, but only $|C|-1$ labels. The proof proceeds by a case distinction on the number of positive coefficients of variables with $x^\ell_i=1$ for a given labeling:
\begin{enumerate}
\item There is only one variable, say $x^{\ell^*}_{i^*}$, with positive coefficient, i.e., all other vertices in the clique are labeled with labels from $L$. Thus $\eqref{eq:cliquen}$ measure exactly the distance from $i^*$ to the "nearest" vertex in the clique, which is a valid upper bound for $b$. 
\item All variables have positive coefficient, i.e., none of the vertices in $C$ gets a label from $L$. In this case, a similar triangle-inequality-based argument as for \eqref{eq:objn} holds, as for the labels $\ell_i, \ell_{i'}$ of at least one edge $\{i,i'\}$ involved in the clique $C$, it must hold that the label $\ell' \in L$ inducing $d(\ell_i,L)$ and $d(\ell_{i'},L)$ must be the same (due to $|C|=|L|+1$).  
\item More than one, but not all variables have positive coefficient, i.e., between one and $|C|-2$ vertices in $C$ gets a label from $L$. We make an additional case distinction.
\begin{enumerate}
\item First, suppose for one of the variables $x^\ell_i$ (corresponding to vertex $i$ with label $\ell_i$) with positive coefficient, distance $d(\ell_i,L)$ gets induced by a label of a vertex in $C$. Thus, the inequality measures at least the distance from vertex $i$ to the "nearest" vertex in the clique similar to case 1 of this proof.
\item Next, suppose for none of the variables $x^\ell_i$ with positive coefficient the distance $d(\ell_i,L)$ gets induced by a label of a vertex in $C$. Let $C^+$ be the vertices in $C$ with positive coefficient and let $C'$ the remaining vertices in $C$ (i.e., the ones with labels in $L$). Let $L'$ be the set of labels after removing from $L$ all the labels of vertices in $C'$, we have that $|L'|=|L|-|C'|=|C|-1-|C'|=|C^+|-1$. As by assumption of this subcase, for each vertex $i \in C^+$, the coefficient in the inequality gets induced by $d(\ell_i,L')$, we are now in a similar case to case 2 of this proof.
\end{enumerate}
\end{enumerate}
\qed
\end{proof}

Similar to \eqref{eq:objn}, the coefficients in \eqref{eq:cliquen} can be downlifted using $UB$. Following is another set of valid inequalities. 

\begin{observation}
	Let $i \in V$ and $\ell \in L$ and $d\in \mathbb N$ a given distance. Then inequalities 
	\begin{equation}
	b \leq |V| + (|V|-d)-(|V|-d)\big(\sum_{\ell':d(\ell,\ell')\leq d}x^{\ell'}_i+\sum_{i':\{i,i'\}\in E}x^\ell_{i'} \big) \label{eq:valid2}\tag{VERTEX-N}
	\end{equation}
	are valid for $(F)$.
	
\end{observation}

\begin{proof}
The sum of the $x$-variables in \eqref{eq:valid2} can be at most two, as both the first and second sum can be at most one. It is easy to see that if the sum is zero or one, the inequality is valid, since the right-hand-side (rhs) in these cases is $|V| + (|V|-d)$, resp., $|V|$. In case the sum of the $x$-variables is two, the rhs is $d$ and thus the inequality reads $b\leq d$, i.e., the maximal antibandwidth of the labeling induced by these $x$-variables is at most $d$. As a sum of two for the $x$-variables implies, that one of the vertex adjacent to $i$ must have label $\ell$ (due to the second sum), and also that vertex $i$ must have a label $\ell'$ within distance $d$ of $\ell$ (due to the first sum), this is a correct estimation.
\qed
\end{proof}	

If an upper bound $UB$ is known, $|V|$ in \eqref{eq:valid2} can be downlifted to $UB$ and only $d < UB$ have to be considered.

\subsection{New formulation $(F_E)$ and an iterative MIP approach}

We now present an extended formulation denoted as $(F_E)$. Let binary variable $b_\ell$, $\ell \in \{1,\ldots, |V|\}$ be one, if and only if the antibandwidth of a solution is $\ell$. The \ABP\ can be modeled as 

\begin{align}
\max \quad & \sum_{1\leq \ell \leq |V|} \ell b_\ell & \notag \\
\eqref{eq:a-sum} & , \eqref{eq:b-sum} & \notag \\
\sum_{1\leq \ell \leq |V|} b_\ell = 1 \label{eq:obje} \tag{OBJ-E} \\
\sum_{\ell_1 < \ell' \leq |V|}  b_{\ell'} + \sum_{\ell_2\leq \ell' \leq \ell_2+\ell_1}(x^{\ell'}_i+x^{\ell'}_i) & \leq 2 & \quad \forall \{i,i'\}\in E, \forall \ell_1 \in \{1,\ldots, |V|\}, 1\leq \ell_2 \leq |V|-\ell_1 \label{eq:obje2} \tag{OBJ-E2} \\
x^\ell_i & \in\{0,1\} & \forall i \in V, \forall \ell \in \{1,\ldots, |V| \} \notag \\
b_\ell  & \in\{0,1\} & \forall \ell \in \{1,\ldots, |V| \} \notag
\end{align}

Constraint \eqref{eq:obje} ensures, that only one variable $b_\ell$ is one, while constraints \eqref{eq:obje2} make sure that the correct variable $b_\ell$, which is compatible with the solution encoded by the $x^\ell_i$-variables is selected: If for an edge $x^\ell_i$-variables, which are within distance $\ell_1$ are one, the constraints ensure that only variables $b_{\ell'}$ with $\ell'<\ell_1$ can be set to one. There are $O(|V|^2)$ variables and $O(|V|^2|E|)$ constraints. Given valid upper and lower bounds $UB$ and $LB$ for the problem, one can remove all variables $b_\ell$ with $\ell>UB$ and $\ell<LB$ and the associated constraints \eqref{eq:obje2}.
However, the resulting MIP is still very large. Thus, we do not use formulation $(F_E)$ directly to solve ABP, but instead, use the following related MIP $(F_E(k))$, which is a feasibility problem, which can be derived to answer the question "Does there exist a solution with $AB(G)\geq k+1$". 

\begin{align}
\max \quad & 0 & \notag \\
\eqref{eq:a-sum} & , \eqref{eq:b-sum} & \notag \\
 \sum_{\ell_2\leq \ell' \leq \ell_2+k}(x^{\ell'}_i+x^{\ell'}_i) & \leq 1 & \quad \forall \{i,i'\}\in E, 1\leq \ell_2 \leq |V|-k \label{eq:objk} \tag{OBJ-k} \\
x^\ell_i & \in\{0,1\} & \forall i \in V, \forall \ell \in \{1,\ldots, |V| \} \notag
\end{align}

The formulation has $O(|V|^2)$ variables and $O(|V||E|)$ constraints. For the given $k$, the set-packing constraints \eqref{eq:objk} ensure that for every edge $\{i,i'\}$, in any feasible solution, the endvertices $i$ and $i$' cannot get labels $f_i, f_{i'}$ which would result in $AB_f\leq k$ for this edge. Thus, any feasible solution to $(F_E(k))$ gives a labeling $f$ with $AB_f \geq k+1$ and also any labeling $f$ with $AB_f\geq k+1$ is a feasible solution for this MIP. Hence, if for a given $k$, $(F_E(k))$ is infeasible, than there is no labeling with $AB_f\geq k+1$. 
Based on $(F_E(k))$, the following simple iterative algorithm to solve \ABP\ can be designed:
\begin{enumerate}
	\item $k \gets 1$
	\item solve $(F_E(k))$
	\item if $(F_E(k))$ is feasible, increase $k$ by one and go back to Step 2
	\item output $k$
\end{enumerate}

In Step 3, instead of increasing $k$ by just one, the antibandwidth of the labeling induced by the solution of $(F_E(k))$ can be used and $k$ for the next iteration can be set to this antibandwidth plus one. Moreover, if a feasible labeling (e.g., obtained by a heuristic) is available, Step 1 can of course start with the value induced by this labeling and not with one.

Similar to $(F)$, the constraints \eqref{eq:objk} of the formulation are actually a special case of a more general family of constraints.
In a first generalization step, we obtain following set of conflict constraints, for which validity follows from their definition and the fact, that every vertex gets exactly one label.

\begin{observation}
For an edge $\{i,i'\}$, let $L_i$ and $L_{i'}$ be two sets of (potential) labels, such that for any $\ell \in L_i$ and $\ell' \in L_{i'}$, we have $|\ell-\ell'| \leq k$. Then inequalities
\begin{equation}
 \sum_{\ell \in L_i}x^{\ell}_i+ \sum_{\ell \in L_{i'}}x^{\ell'}_{i'} \leq 1 \label{eq:con} \tag{CONFLICT}
\end{equation}
are valid for $(F_E(k))$.
\end{observation}

The inequalities \eqref{eq:con} can be further generalized as follows using the same conflict-arguments, i.e., if two of the variables would be one, the solution would induce a labeling $f$ with $AB_f \leq k$.

\begin{observation}
Let $C\subseteq V$ be a set of vertices forming a clique in $G$ and let $L_c$ be sets of labels, one for each $c \in C$,
such that for any $\ell \in L_c$ and $\ell' \in L_{c'}$ for $c, c' \in C$, we have $|\ell-\ell'| \leq k$. Then inequalities
\begin{equation}
 \sum_{c \in C} \sum_{\ell \in L_c}x^{\ell}_c \leq 1 \label{eq:cli} \tag{CLIQUE-E}
\end{equation}
are valid for $(F_E(k))$.
\end{observation}

The formulation $(F_E(k))$ and in particular inequalities \eqref{eq:cli} show the strong relation of the \ABP\ to the SSP: In the SSP, any feasible solution is just allowed to take one vertex for each clique in the graph. In the \ABP, for a labeling $f$ with $AB_f\geq k+1$ to exist, any set of vertices getting labels within a distance of $k$ must be a stable set. Thus, in particular, for each clique in the graph and a set of labels $L$ where all labels in $L$ are within a distance of $k$, exactly one vertex in the clique can be given a label of set $L$, which is exactly what \eqref{eq:cli} enforces.

\subsection{Implementation details \label{sec:details}}

In this section, we discuss implementation details of the branch-and-cut algorithm we developed based on $(F)$ and the iterative algorithm based on $(F_E(k))$ (in which the individual problems $(F_E(k))$ for a fixed $k$ also get solved with a branch-and-cut algorithm). While both formulations are compact (i.e., have a polynomial number of constraints), the number of constraints is still very large, and the constraints also are very dense (i.e., have many non-zero coefficients). Thus, we do not add them all in the beginning, but separate them on-the-fly when they are violated. Moreover, we also separate inequalities \eqref{eq:valid2}, \eqref{eq:cliquen}., resp. \eqref{eq:cli}, details are given in the following. We set the limit for separation-rounds to twenty at the root-node and to one at all the other nodes in the branch-and-cut tree to avoid overloading the LPs with too many inequalities. 
The coefficients of all inequalities used when solving $(F)$ are downlifted using the best upper bound obtained by applying Theorem \ref{thm:bounds} and we also use this upper bound as the termination criterion for the iterative algorithm (i.e., thus we solve the SSP and GCP before starting our algorithms). In both approaches, we initialize the MIP-model with just constraints \eqref{eq:a-sum}, \eqref{eq:b-sum} and the symmetry breaking discussed below.

\paragraph{Symmetry breaking}

As the objective function of the problem uses the absolute value, any labeling and its reversed version give the same objective function value. Thus, to break these symmetries, we add constraints such that the vertex with maximum degree in the graph must have a label at most $\lceil|V|/2\rceil$ (if there is more than one vertex with maximum degree, we take the one with smallest index). We do this by fixing the corresponding non-allowed $x^\ell_i$-variables to zero.

\paragraph{Separation routine for $(F)$}

We do different separation routines depending on whether the current solution$(\tilde x, \tilde d)$\sloppy\ to the LP-relaxation at the current branch-and-cut node is integral or not. If the solution is integral, we simply check by inspection, if any of the constraints \eqref{eq:objn} is violated, and add any violated constraints. Note that this would already be enough to ensure correctness of the branch-and-cut (CPLEX, the MIP-solver we used, also produces integral solution with its internal heuristics, we also check these solutions in a similar fashion).

Given a fractional solution $(\tilde x, \tilde d)$, we first try to find violated inequalities \eqref{eq:valid2}. This is done by enumeration, and we add at most one violated inequality for each vertex, i.e., when we found a violated inequality for a vertex, we stop enumeration for this vertex and move to the next one.

If the previous procedure does not produce any violated inequalities, we try a heuristic separation of inequalities \eqref{eq:cliquen}. We note that compared to separation of clique-type inequalities in other problems such as e.g, the SSP, in our case we also need to find a set of labels to define the inequalities. We thus first compute a \emph{pseudoposition} $p_i$ induced by the current fractional $\tilde x$ for each vertex $i$ by $p_i=\sum_{\ell \in L} \ell \tilde x^\ell_i$. We then iterate over each edge $e =\{i,i'\} \in E$ and greedily try to construct a violated inequality  \eqref{eq:cliquen} containing this edge by iteratively adding more vertices, which form a clique. More precisely, for an $e =\{i,i'\}$, our initial $C=\{i,i'\}$ and to increase $C$, we take all vertices, which are adjacent to all vertices in $C$ as candidate set $C'$. We then rank each $i'' \in C'$ by calculating $score_{i''}=\sum_{i \in C} |p_i-p_{i''}|$. The vertex with minimal score gets added to $C$ and the procedure gets repeated, until $C'=\emptyset$, i.e., there exists no vertex to further grow the clique $C$. With this approach, we try to find cliques $C$, where the vertices have labels which are near to each other, as such a labeling would induce a small value of $b$ and thus hopefully leads to a violated inequality. To specify an inequality \eqref{eq:cliquen} for a given clique $C$, we also need a set of $|C|-1$ labels. For this, we calculate $labelscore_\ell= \sum_{c \in C} \tilde x^\ell_c$ and take the  $|C|-1$ labels with the highest score. Whenever a violated inequality is found, we mark all the edges in the corresponding clique, and we do not consider marked edges for the remainder of the separation procedure.

Finally, if also no violated inequalities \eqref{eq:cliquen} were found, we try a partial enumeration to heuristically separate inequalities \eqref{eq:objn}: For each edge $e =\{i,i'\} \in E$, we check, if the inequality \eqref{eq:objn} for $\ell$ with maximum $\tilde x^\ell_i+\tilde x^\ell_{i'}$ is violated.

\paragraph{Separation routine for $(F_E(k))$}


Inequalities \eqref{eq:objk} separated by enumeration. Once a violated inequality for an edge $e =\{i,i'\} \in E$ is found, we try to lift it to a clique inequality \eqref{eq:cli} using an iterative heuristic (as set of labels for each vertex in \eqref{eq:cli}, we consider $\ell \in [\ell_2, \ell_2+k]$ , where $\ell_2$ is the label defining the violated inequality \eqref{eq:objk}). We initialize $C$ with $\{i,i'\}$, and consider as candidate vertices $C'$ for lifting all vertices adjacent to $C$. For each of these vertices $i'' \in C'$ we calculate a score $\big(\sum_{\ell_2 \leq \ell \leq \ell_2+k} \tilde x^\ell_{i''}+\epsilon \big) \cdot \delta(i)$, where $\epsilon=0.0001$. The vertex with the biggest score is added to $C$, and the process is repeated, until there is no more vertex available to increase $C$. Similar to the separation of clique inequalities \eqref{eq:cliquen}, once an edge occurs in an added inequality, it is not considered anymore in the remainder of the separation procedure.

%

\paragraph{Branching} During the branch-and-cut, the branch-and-cut trees can become very unbalanced, as branching on an $x^\ell_i$-variable fixes a vertex to a label in one branch, and forbids this label for this vertex in the other branch, while all other (not previously fixed) labels are still possible for this vertex. We thus implemented our own branching strategy. Given the solution $(\tilde x)$ of an LP-relaxation at a node, we consider all vertices $i$, where the subvector $(\tilde x_i)$ has fractional entries as branching candidates. Among these vertices, we take the one with the highest degree to branch on. If there is more than one candidate, we take the one with the largest number of fractional entries in the subvector $(\tilde x_i)$, if there are still ties we break them arbitrarily, i.e., we take the vertex with the smallest index. Regarding the branching itself, we do not branch on a single label, but branch on $\sum_{\ell' < \ell} x^{\ell'}_i$ for a given label $\ell$. In one branch, this sum must be zero, and in the other branch, this sum must be one. The label $\ell$ is determined as the largest $\ell'$ with $\tilde x^{\ell'}_i>0$.

\paragraph{Starting heuristics and primal heuristic}

We implemented two starting heuristics to create an initial starting solution, and also a primal heuristic which is called during the branch-and-cut and guided by the value of the LP-solution at the current branch-and-cut node. All three heuristics consecutively iterate over the labels in an increasing way, starting at label $1$, and give each label to an yet unlabeled vertex, which is then removed for consideration for the remaining labels (i.e., there is no label-reassignment during the heuristics). 

The first starting heuristic is in similar spirit to construction heuristics used in e.g., \cite{bansal2011memetic,duarte2011grasp,lozano2012variable}. Given a vertex $i^* \in V$, we construct a breadth-first-search (bfs) tree $T_{i^*}$ starting from $i^*$. This tree has layers $T_{i^*}(k)$, $k\geq 0$, where layer $T_{i^*}(k)$ contains all vertices $i'$ with $k-1$ vertices on the path between it and ${i^*}$ in the tree (e.g., $T_{i^*}(1)$ contains all vertices adjacent to $i^*$), and we define $T_{i^*}(0)=i^*$. Note that by construction of a bfs-tree, adjacent vertices in $G$ are either on the same layer or in two consecutive layers $k,k+1$. Naturally, to get a large antibandwidth, we do not want to give adjacent vertices labels which are close to each other. Thus, we give label one to vertex ${i^*}$ and then repeatedly iterate through the even and odd layers of $T_{i^*}$ to assign the remaining labels to vertices. By switching between even and odd layers, we try to avoid giving close labels to vertices which are adjacent and in consecutive layers in the tree. However, vertices on the same layer may also be adjacent in $G$. Thus, whenever a vertex $i'$ gets assigned a label, we mark all vertices adjacent to $i'$ and we do not consider marked vertices for assigning labels in the current iteration. The order in which we consider the vertices within a layer for assigning labels is induced by the following three criteria: i) resulting antibandwidth for this vertex if the vertex gets assigned the currently considered label (all unlabeled vertices are defined to have label $|V|-1$ for this calculation), ii) degree of the vertex in the graph consisting of the yet unlabeled vertices, iii) maximum degree of an adjacent vertex in the graph consisting of the yet unlabeled vertices. The vertices in a layer are ordered in descending order according to i), ties are first broken by descending order according to ii), if there still remain ties, they are broken by descending order according to iii), the remaining ties are broken arbitrarily, i.e., the vertex with the smallest index is taken. 

In the second heuristic, we keep a bound $B^H$, which we initialize with the best $UB$ according to Theorem \ref{thm:bounds}
We start by assigning some given vertex $i^*$ the label one, and the continue assigning the remaining labels. For assigning the currently considered label, we consider all the unlabeled vertices, where assigning the current label would result in an antibandwidth of the vertex with value at least $B^H$ (similar to i) above, unlabeled vertices are defined to have label $|V|-1$ for this calculation). If there is more than one vertex fulfilling this condition, we use criteria ii) and then iii) for tie-breaking. If there is no vertex fulfilling the condition, we decrease $UB^H$ until there is again at least one vertex fulfilling the condition. We run both starting heuristics with all vertices $i \in V$ as $i^*$.

The primal heuristic also uses a bound $B^H$, which get initialized to the value of the current incumbent solution plus one. We again start the labeling with assigning label one, and then proceed to the next label. For assigning any label $\ell$ (including label one), we sort all the unlabeled vertices $i$ in descending order according to $ (\tilde x^\ell_i+\epsilon)\cdot \delta(i)$, where $\epsilon=0.0001$. We iterate through this ordered list of vertices and assign $\ell$ to the first vertex, which fulfills the condition that assigning $\ell$ to it would result in an antibandwidth of the vertex with value at least $B^H$. If there is no vertex fulfilling this condition, we decrease $B^H$ until there is at least one vertex fulfilling it.

Given a solution $f^H$ obtained by any of the heuristics, we try to improve it with an iterative local search procedure. For the solution $f^H$, we calculate the set of edges $minE$, which are all the edges with $|f^H(i)-f^H(i')|=AB_{f^H}$, i.e., the edges with the minimum bandwidth. We then iterate trough all the edges $e\{i,i'\} \in minE$ and try to improve the bandwidth, by switching the labels $f^H(i)$ or $f^H(i')$ with labels of vertices $i'' \neq i,i'$. For each edge, we apply the switch resulting in the largest bandwidth considering both vertices involved in the switch. We update $minE$ and repeat this procedure until no more improvement of the bandwidth is possible. 



\section{A constraint programming formulation \label{sec:cp}}

The \ABP\ can also be straightforwardly modeled as constraint programming (CP) problem (see, e.g., \cite{rossi2006handbook} for more on CP) using the $abs$ and $alldifferent$-constraints.

\begin{align}
\max b& \tag{C.1} \\
b \leq abs(l_i-l_{i'}) & \quad \forall \{i,i'\}\in E &\tag{C.2}\\
alldifferent(l) &\tag{C.3} \\
l_i \in \{1,2,\ldots, |V|\},& \quad \forall i \in V \tag{C.4}
\end{align}

Similar to the MIP-approaches, we also add a symmetry breaking constraint restricting the label of the vertex with maximum degree to be at most $\lceil|V|/2\rceil$ when solving the problem as CP.

\newcommand{\FA}{$(F)$}
\newcommand{\FTWO}{$(F_e(k))$}
\newcommand{\FPAPER}{$(F_{lit})$}
\newcommand{\CP}{CP}

\section{Computational results \label{sec:comp}}

The branch-and-cut framework and the iterative MIP algorithm, as as well as the MIPs for the SSP and GCP were implemented in C++ using CPLEX 12.9 as MIP solver. To solve the CP formulation, we use the CP optimizer of CPLEX.
The runs were carried out on an Intel Xeon E5 v4 CPU with 2.5 GHz and 6GB memory using a single thread, with timelimit for a run set to 1800 seconds. As timelimit of the MIPs for solving the SSP and GCP we set 10 seconds. Note that the calculations for the upper bounds based on SSP and GCP are still valid when $\alpha(G)$ is replaced with an upper bound, and  $\chi(G)$ is replaced with a lower bound, both are available even if the corresponding MIP is not solved to optimality within the given timelimit. All CPLEX parameters were left at their default values, except the choice of the simplex algorithm used within the branch-and-cut. As default, CPLEX would use the \emph{dual simplex} (which allows for faster re-solving after adding constraints), however, as our constraints are very dense, and also the number of variables is much larger as the number of (added) constraints, using the \emph{primal simplex} turned out to be more efficient in preliminary computations (see e.g., \cite{klotz2013practical, klotz2013practical2} for a discussion on the choice of LP-algorithm).

\subsection{Instances}

In our computational study, we focused on the HarwellBoeing instances, which are the main instances used in performance tests for the ABP. These instances are based on the Harwell-Boeing Sparse Matrix Collection, which is a "is a set of standard test matrices arising from problems in linear systems, least squares, and eigenvalue calculations from a wide variety of scientific and engineering disciplines", see \url{https://math.nist.gov/MatrixMarket/collections/hb.html}. The instances are available at and are available at \url{https://www.researchgate.net/publication/272022702_Harwell-Boeing_graphs_for_the_CB_problem} and have also been used for testing algorithms for other labeling problem, see, e.g., \cite{rodriguez2015tabu,sinnl2018}. The set contains instances with up to 715 vertices and 2975 edges, details of the number of vertices and edges of the individual graphs are given in Table \ref{ta:main} in columns $|V|$ and $|E|$. The instances with up to including 118 vertices are denoted as \emph{small}, the remaining ones as \emph{large}, both groups contain twelve instances.
We observe that also other instances have been used in testing, e.g., paths, grids or Hamming graphs in \cite{lozano2012variable}, however, for these specific graphs, optimal solution values are known due to theoretical results.

\subsection{Results}

First, we are interested in the strength of the LP-relaxation of the new model $(F)$ compared to the previous model $(F_{lit})$
In Figure \ref{fig:lpgap}, we give the LP-gaps. The gaps are calculated as $100\cdot (UB_{LP}-z^*)/z^*$, where $UB_{LP}$ is the value of the respective LP-relaxation, and $z^*$ is the value of the best known feasible solution for the instance. For the best solution value, we take the results from Table 6 in \cite{lozano2012variable}, which gives a comparison of the state-of-the-art heuristics from \cite{bansal2011memetic,duarte2011grasp,lozano2012variable}, and also the solution values our algorithms obtained (these results are discussed later in this section in detail). For these runs, we directly solve the LP-relaxation of the compact model $(F)$ (and $(F_{lit})$, which we also implemented) without any lifting of coefficients or valid inequalities. We made runs only for the \emph{small} instances, as for larger ones, solving the compact LPs becomes computationally burdensome. 

\begin{figure}[h!tb]
	\centering
		\includegraphics[width=0.8\textwidth]{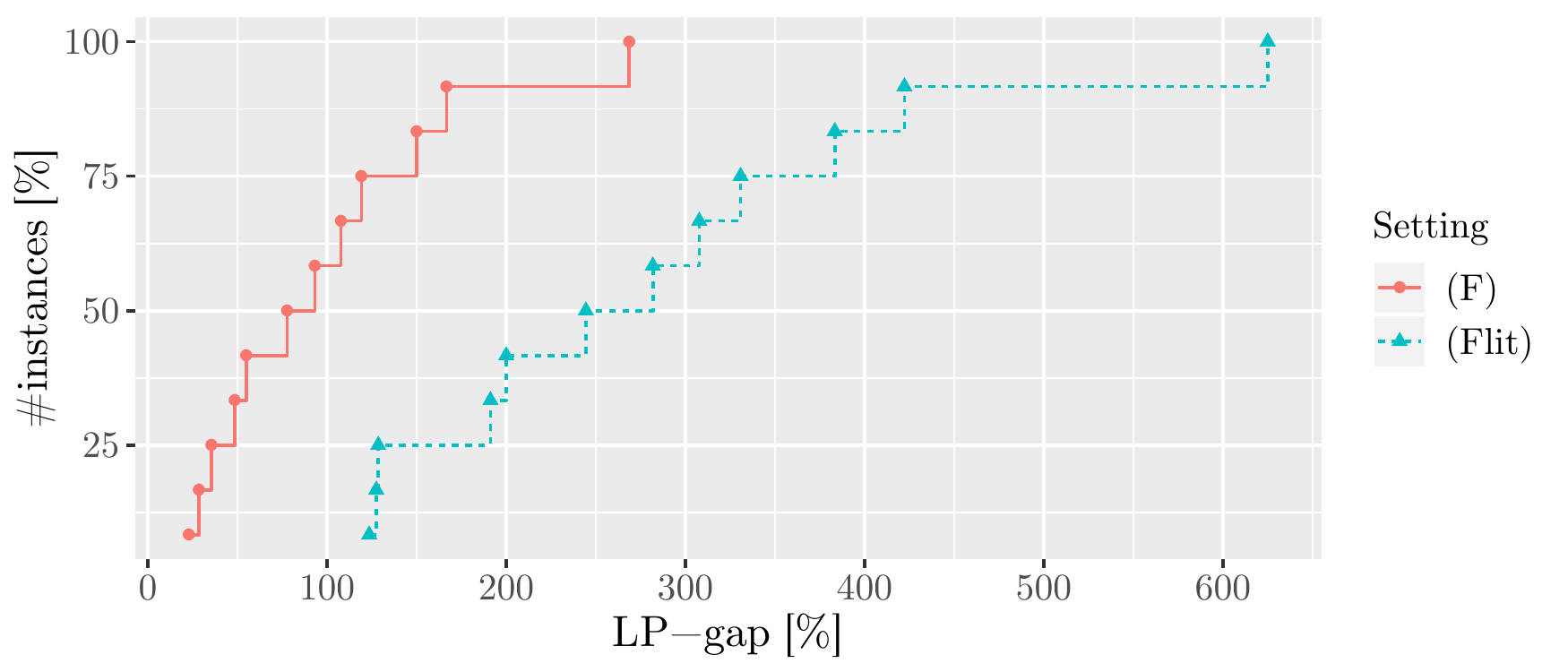}
		\caption{LP-gap of formulations $(F)$ and $(F_{lit})$ with respect to the best known solution for the \emph{small} instances.\label{fig:lpgap}}
\end{figure}

Figure \ref{fig:lpgap} shows that the new model brings a big improvement in the value of the LP-gaps. The gaps of $(F)$ are smaller for all instances, the largest gap for $(F_{lit})$ is over 600\%, while the largest gap for $(F)$ is under 300\%. In general, the gaps of $(F)$ seem about 100-200\% smaller, and more than half of the gaps of $(F)$ are under 100\%.



Next, we come to our main results, which are contained in Table \ref{ta:main}.  In this table, we report for each instance

\begin{itemize}
	\item the number of vertices (column $|V|$)
	\item the number of edges (column $|E|$)
	\item the upper bounds calculated using Theorem \ref{thm:bounds} (columns $T\ref{thm:bounds}.1, T\ref{thm:bounds}.2, T\ref{thm:bounds}.3, T\ref{thm:bounds}.4$), including the time needed for solving the SSP and GCP for calculating $\alpha(G)$ and $\chi(G)$ (columns $t_{\alpha(G)}$ and $t_{\chi(G)}$)
	\item the upper bound $UB_{D}$ reported in \cite{duarte2011grasp} for the MIP $(F_{lit})$. For these runs, the authors of \cite{duarte2011grasp} used CPLEX 12.3 with a timelimit of 24 hours; note that for a fairer comparison, we have also re-implemented $(F_{lit})$ and run it with CPLEX 12.9 and a timelimit of 1800 seconds, and give detailed results for our runs, see below
	\item the best solution value from literature (column $z_{L}$, taken from Table 6 of \cite{lozano2012variable})
	\item the best optimality gap using results from literature (column $g_{L}[\%]$, calculated as $100 \cdot (UB_L-z_L)/z_L$, where $UB_L$ is the best upper bound reported in literature, i.e., the minimum of $T\ref{thm:bounds}.1, T\ref{thm:bounds}.2$ and the upper bound $UB_D$ reported in \cite{duarte2011grasp}. Note that Table 6 of \cite{lozano2012variable} does not consider the latter upper bound, and the percentage deviation they report is calculated as $100 \cdot (UB-z)/UB$ for their upper bounds and solution value
	\item the upper bound $UB$, best solution value $z^*$ and runtime obtained by our approaches $(F)$, $CP$, as well as by our re-implementation of $(F_{lit})$ using CPLEX 12.9. For $(F_e(k))$, we report just the best solution value and the runtime, as this algorithm gives no upper bound, except when it manages to prove optimality (which is indicated by a runtime smaller than the timelimit in the table).
	\item the best optimality gap $g^*$ obtained after taking into account our new results. The gap is calculated as $100 \cdot (UB_B-z_B)/z_B$, where $UB_B$ and $z_B$ are the best available upper bounds, resp., best solution values taking into account our results, and also the best solution values previously reported in literature (i.e., as given in column $z_L$).
\end{itemize}

\begin{landscape}
	\setlength{\tabcolsep}{3pt}
\begin{table}[ht]
	\centering
	\caption{Detailed results, the horizontal line shows the difference between the \emph{small} and \emph{large} instances. Best entries for $UB$ and best solution value of each instance are given in bold. \label{ta:main}} 
	\begingroup\footnotesize
	\begin{tabular}{lrr|rrrrr|rrrr|rrr|rrr|rr|rrr|r}
		\toprule
		& & & T1.1 & T1.2 & & & & \multicolumn{2}{|c}{T1.3} & \multicolumn{2}{c}{T1.4}  & \multicolumn{3}{|c|}{\FPAPER} & \multicolumn{3}{|c|}{\FA} & \multicolumn{2}{|c|}{\FTWO} & \multicolumn{3}{|c|}{\CP} &   \\ name & $|V|$ & $|E|$ & UB & UB & UB$_D$ & $z_L$ & $g_L$ & UB & t[s] & UB & t[s] & UB & $z^*$ & t[s] & UB & $z^*$ & t[s] & $z^*$ & t[s] & UB & $z^*$ & t[s] & $g^*$ \\ \midrule
		pores1 & 30 & 103 & 13 & 16 & 8 & \textbf{6} & 33.3 & 8 & 1 & 9 & 1 & 8 & \textbf{6} & TL & 8 & \textbf{6} & TL & \textbf{6} & 24 & 8 & \textbf{6} & TL & 0.0 \\ 
		ibm32 & 32 & 90 & 15 & 19 & \textbf{9} & \textbf{9} & 0.0 & 13 & 1 & 10 & 1 & \textbf{9} & \textbf{9} & 73 & \textbf{9} & \textbf{9} & 280 & \textbf{9} & 27 & \textbf{9} & \textbf{9} & 2 & 0.0 \\ 
		bcspwr01 & 39 & 46 & 19 & 29 & \textbf{17} & \textbf{17} & 0.0 & 21 & 1 & 19 & 1 & \textbf{17} & \textbf{17} & 138 & 18 & \textbf{17} & TL & \textbf{17} & 7 & \textbf{17} & \textbf{17} & 1 & 0.0 \\ 
		bcsstk01 & 48 & 176 & 22 & 29 & 11 & 8 & 37.5 & 13 & 1 & \textbf{9} & 1 & 22 & 8 & TL & \textbf{9} & 8 & TL & \textbf{9} & 29 & \textbf{9} & \textbf{9} & 4 & 0.0 \\ 
		bcspwr02 & 49 & 59 & 24 & 38 & 22 & \textbf{21} & 4.8 & 27 & 1 & 24 & 1 & 22 & \textbf{21} & TL & 23 & \textbf{21} & TL & \textbf{21} & 629 & \textbf{21} & \textbf{21} & 5 & 0.0 \\ 
		curtis54 & 54 & 124 & 26 & 38 & \textbf{13} & \textbf{13} & 0.0 & 22 & 1 & \textbf{13} & 1 & \textbf{13} & \textbf{13} & 814 & \textbf{13} & 12 & TL & \textbf{13} & 5 & \textbf{13} & \textbf{13} & 9 & 0.0 \\ 
		will57 & 57 & 127 & 28 & 41 & 14 & \textbf{13} & 7.7 & 25 & 1 & 14 & 1 & 18 & 12 & TL & 14 & \textbf{13} & TL & \textbf{13} & 10 & \textbf{13} & \textbf{13} & 21 & 0.0 \\ 
		impcolb & 59 & 281 & 29 & 35 & 14 & \textbf{8} & 75.0 & 21 & 1 & \textbf{8} & 1 & 19 & \textbf{8} & TL & \textbf{8} & \textbf{8} & 4 & \textbf{8} & 1 & 22 & \textbf{8} & TL & 0.0 \\ 
		ash85 & 85 & 219 & 42 & 64 & \textbf{27} & 21 & 28.6 & 29 & 1 & 28 & 1 & 41 & 17 & TL & 28 & 19 & TL & 20 & TL & 32 & \textbf{22} & TL & 22.7 \\ 
		nos4 & 100 & 247 & 50 & 78 & 47 & \textbf{34} & 38.2 & \textbf{40} & 1 & 49 & 1 & 49 & 27 & TL & \textbf{40} & 32 & TL & 32 & TL & 47 & \textbf{34} & TL & 17.6 \\ 
		dwt234 & 117 & 162 & 58 & 99 & 58 & 50 & 16.0 & 76 & 1 & 58 & 1 & 58 & 46 & TL & 58 & 46 & TL & 49 & TL & \textbf{57} & \textbf{51} & TL & 11.8 \\ 
		bcspwr03 & 118 & 179 & 59 & 99 & 57 & \textbf{39} & 46.2 & 57 & 1 & \textbf{39} & 1 & \textbf{39} & \textbf{39} & 874 & \textbf{39} & \textbf{39} & 9 & \textbf{39} & 1 & \textbf{39} & \textbf{39} & 1 & 0.0 \\ 
		\midrule
		bcsstk06 & 420 & 3720 & 210 & 334 & 210 & 32 & 556.2 & 72 & 1 & \textbf{38} & 5 & 343 & 1 & TL & 71 & 29 & TL & \textbf{33} & TL & 186 & 30 & TL & 15.2 \\ 
		bcsstk07 & 420 & 3720 & 210 & 334 & 210 & 31 & 577.4 & 72 & 1 & \textbf{38} & 5 & 343 & 1 & TL & 71 & 29 & TL & \textbf{33} & TL & 186 & 30 & TL & 15.2 \\ 
		impcold & 425 & 1267 & 212 & 375 & 212 & 103 & 105.8 & 173 & 2 & \textbf{141} & 1 & 353 & 7 & TL & 425 & 91 & TL & 99 & TL & 195 & \textbf{110} & TL & 28.2 \\ 
		can445 & 445 & 1682 & 221 & 387 & 221 & \textbf{82} & 169.5 & \textbf{120} & 1 & 148 & TL & 407 & 6 & TL & 445 & 78 & TL & 78 & TL & 217 & 74 & TL & 46.3 \\ 
		494bus & 494 & 586 & 247 & 460 & 247 & \textbf{227} & 8.8 & 278 & 1 & \textbf{246} & 1 & 391 & 24 & TL & 494 & 219 & TL & 219 & TL & \textbf{246} & 217 & TL & 8.4 \\ 
		dwt503 & 503 & 2762 & 250 & 429 & 250 & 53 & 371.7 & 127 & 1 & \textbf{71} & 7 & 460 & 2 & TL & 503 & 46 & TL & 51 & TL & 246 & \textbf{56} & TL & 26.8 \\ 
		sherman4 & 546 & 1341 & 272 & 494 & \textbf{272} & \textbf{261} & 4.2 & 273 & 1 & 545 & 1 & 545 & 1 & TL & 546 & 256 & TL & 256 & TL & 543 & 211 & TL & 4.2 \\ 
		dwt592 & 592 & 2256 & 295 & 525 & 295 & \textbf{113} & 161.1 & \textbf{150} & 1 & 197 & TL & 492 & 6 & TL & 592 & 103 & TL & 103 & TL & 275 & 99 & TL & 32.7 \\ 
		662bus & 662 & 906 & 331 & 619 & 331 & \textbf{220} & 50.5 & 351 & 1 & \textbf{220} & 1 & 660 & 1 & TL & 347 & 219 & TL & 219 & TL & \textbf{220} & \textbf{220} & 193 & 0.0 \\ 
		nos6 & 675 & 1290 & 337 & 624 & \textbf{337} & \textbf{329} & 2.4 & 338 & 1 & 674 & 1 & 674 & 1 & TL & 675 & 326 & TL & 326 & TL & 672 & 271 & TL & 2.4 \\ 
		685bus & 685 & 1282 & 342 & 634 & 342 & \textbf{136} & 151.5 & 313 & 1 & \textbf{136} & 1 & 621 & 5 & TL & 242 & \textbf{136} & TL & \textbf{136} & 1 & 342 & \textbf{136} & TL & 0.0 \\ 
		can715 & 715 & 2975 & 357 & 638 & 357 & \textbf{115} & 210.4 & 208 & 1 & \textbf{142} & TL & 649 & 6 & TL & 242 & 112 & TL & 112 & TL & 333 & 112 & TL & 23.5 \\ 
		\bottomrule
	\end{tabular}
	\endgroup
\end{table}

\end{landscape}

Several interesting results can be seen in Table \ref{ta:main}.
The best known solution values from literature are already quite good, however, for seven instances, we were able to find better values. Five of these improvements were achieved by approach CP, and three by $(F_e(k))$ (for one instance, both managed it). This is a strong contrast to the statement "On the contrary, the [\ldots] HarwellBoeing instances are actually a challenge for modern heuristic methods" in \cite{lozano2012variable}. Our results reveal, that the large gaps were mostly caused by bad upper  bounds. Indeed, for none of the instances, bounds $T\ref{thm:bounds}.1$ or $T\ref{thm:bounds}.2$ (as used in \cite{lozano2012variable} to assess the quality of the best obtained solution values) are among the best bounds, they are often very far away. Also, the $UB_D$ of the MIP $(F_{lit})$ as reported in \cite{duarte2011grasp} is similar to the best known upper bound after our current study for only six instances. Amongst the bounds provided by Theorem \ref{thm:bounds}, the bound based on the chromatic number $\chi(G)$ seems to be the strongest, for twelve instances, it provides the (sometimes jointly) best upper bounds. For six instances, the provided upper bound is actually the same as the best solution value, which proves optimality. The associated NP-hard GCP can be solved within the given timelimit of ten seconds for all but three instances. 

Using our approaches, the optimal solution value is now known for eleven instances (nine out of the twelve \emph{small} instances, and for two out of twelve \emph{large} instances), compared to three instances before. The optimality gaps are improved for further eleven instances, the only (previously unsolved) instances, where we were not able to achieve any improvement are \texttt{sherman4} and \texttt{nos6} (for which the gaps are already quite small with 4.2\% and 2.4 \%). The largest optimality gap is now 46.3\% (for instance \texttt{can445}), compared to 577.4\% before (for instance \texttt{bcsstk07}). 

Regarding effectiveness of our approaches, both $(F_k(e))$ and CP work much better than $(F)$. Surprisingly, although $(F)$ has better LP-gaps compared to $(F_{lit})$, it actually gives worse performance in solving the problem compared to our re-implementation of $(F_{lit})$. This could be caused by the fact that $(F_{lit})$ is sparser than $(F)$, so LP-solving is faster and more nodes can be enumerated. We can also note that this re-implementation is more effective within our timelimit of 1800 seconds, compared to the runs made in \cite{duarte2011grasp} with a timelimit of 24 hours. This is likely caused by the improvements in CPLEX from version 12.3 to 12.9, and also be the better computer we used in our runs. The approach $(F_k(e))$ manages to prove optimality for ten instances within the timelimit, while CP manages to do so for eight instances, our re-implementation of $(F_{lit})$ for four, and $(F)$ for three. 

\section{Conclusions \label{sec:concl}}

In this note, we considered the antibandwidth problem (ABP) and provided improved upper and lower bounds for standard benchmark instances from literature, for which the optimality gaps were up to 577\%. We presented new MIP-formulations for the model, and designed a branch-and-cut algorithm and an iterative solution algorithm based on them. We also developed a constraint programming approach and calculated bounds using the NP-hard stable set problem and graph coloring problem. In a computational study, we showed that the developed approaches allow to find the proven optimal solution for eight instances from literature (out of a commonly used set of 24 benchmark instances for this problem), where the optimal solution was unknown and also provide reduced gaps for eleven additional instances, including improved solution values for seven instances. The largest gap is now 46\%. 

There are several avenues for further work: Trying to improve the MIP-approaches is a possibility, however, for larger-scale instances, the size of MIP-models seems to become prohibitive for solving the problem. But exploring the MIP-approaches further, especially the connection to the stable set problem could maybe be interesting with respect to theoretical results, e.g., finding complete descriptions for certain graph classes. Moreover, as there is also some connection of the ABP to the graph coloring problem, trying to develop column generation/branch-and-price approaches could be a worthwhile topic, as for graph coloring problems, such approaches usually work quite well.
%
Aside from using a MIP-based approach directly to solve the problem, a combinatorial branch-and-bound algorithm could be an interesting idea. Such approaches, which typically work with \emph{partial labelings} and use problem-specific, graph-theoretic bounds for the problem at hand often are often quite effective for graph labeling problems \citep{caprara2005laying,marti2010branch,marti2013branch}. In case of the ABP, even solving NP-hard problems within the branch-and-bound to provide bounds could be a viable option, since our computational study showed, that the stable set problem and graph coloring problem can be solved very quickly for the standard benchmark instances of the ABP.

\section*{Acknowledgements}

The research was supported by the Austrian
Research Fund (FWF, Project P 26755-N19 and P 31366-NBL). The author wants to thank Georg Brandst\"atter for interesting conversations about the problem.



\end{document}